\documentclass[a4paper,UKenglish,cleveref, autoref]{lipics-v2021}
\hideLIPIcs

\usepackage{amsfonts, amsmath, amssymb, amsthm}
\usepackage{hyperref}
\usepackage{verbatim}
\usepackage[noend]{algpseudocode}
\usepackage{xspace,enumerate}
\usepackage{caption}
\usepackage{subcaption}
\usepackage{cancel}

\usepackage{tikz}
\usetikzlibrary{calc}

\usepackage[short,nocomma]{optidef}
\usepackage{cases}
\usepackage{nicefrac}
\usepackage[utf8]{inputenc}
\usepackage{multirow}
\usepackage{url}
\usepackage{mathtools}
\usepackage[utf8]{inputenc}
\usepackage{color}
\usepackage{scalerel}
\usepackage{todonotes}
\theoremstyle{plain}
\newtheorem{fact}[theorem]{Fact}

\def\dd{\mathinner{.\,.}}
\newcommand{\cO}{\mathcal{O}}
\newcommand{\Oh}{\mathcal{O}}
\newcommand{\Minimizers}{\mathcal{M}_{w,k,\rho}}
\newcommand{\AL}{\textsf{Prepend}\xspace}
\newcommand{\AR}{\textsf{Append}\xspace}
\newcommand{\RL}{\textsf{DeleteFirst}\xspace}
\newcommand{\RR}{\textsf{DeleteLast}\xspace}

\nolinenumbers

\title{Minimizers in Semi-Dynamic Strings}
\author{Wiktor Zuba}{University of Warsaw, Poland \and CWI, Amsterdam, The Netherlands}{w.zuba@mimuw.edu.pl}{https://orcid.org/0000-0002-1988-3507}{Supported by the Netherlands Organisation for Scientific Research (NWO) through Gravitation-grant NETWORKS-024.002.003.}
\author{Oded Lachish}{Birkbeck, University of London, London, UK}{o.lachish@bbk.ac.uk}{https://orcid.org/0000-0001-5406-8121}{}
\author{Solon P. Pissis}{CWI, Amsterdam, The Netherlands\and Vrije Universiteit, Amsterdam, The Netherlands}{solon.pissis@cwi.nl}{https://orcid.org/0000-0002-1445-1932}{Supported in part by the PANGAIA and ALPACA projects that have received funding from the European Union’s Horizon 2020 research and innovation programme under the Marie Skłodowska-Curie grant agreements No 872539 and 956229, respectively.}

\Copyright{Wiktor Zuba, Oded Lachish, Solon P. Pissis} 


\funding{A research visit during which part of the presented ideas were conceived was funded by a Royal Society International Exchanges Award.}

\EventEditors{}
\EventNoEds{1}
\EventLongTitle{}
\EventShortTitle{}
\EventAcronym{}
\EventYear{}
\EventDate{}
\EventLocation{}
\EventLogo{}
\SeriesVolume{}
\ArticleNo{}

\authorrunning{W. Zuba et al.}

\begin{document}
\maketitle

\begin{abstract}
Minimizers sampling is one of the most widely-used mechanisms for sampling strings.
Let $S=S[0]\ldots S[n-1]$ be a string over an alphabet $\Sigma$.
In addition, let $w\geq 2$ and $k\geq 1$ be two integers and $\rho=(\Sigma^k,\leq)$ be a total order on $\Sigma^k$.
The minimizer of window $X=S[i\dd i+w+k-2]$ is the smallest position in $[i,i+w-1]$ where the smallest length-$k$ substring of $S[i\dd i+w+k-2]$ based on $\rho$ starts.
The set of minimizers for all $i\in[0,n-w-k+1]$ is the set $\Minimizers(S)$ of the minimizers of $S$. The set $\Minimizers(S)$ can be computed in $\cO(n)$ time. The folklore algorithm for this computation
computes the minimizer of every window in $\cO(1)$ amortized time using $\cO(w)$ working space.

It is thus natural to pose the following two questions: 
\begin{center}
\textbf{Question 1}: \emph{Can we efficiently support other dynamic updates on the window?}
\end{center}
\begin{center}
\textbf{Question 2}: \emph{Can we improve on the $\cO(w)$ working space?}
\end{center}
\noindent We answer both questions in the affirmative:

\begin{enumerate}
    \item We term a string $X$ \emph{semi-dynamic} when one is allowed to insert or delete a letter at any of its ends. We show a data structure that maintains a semi-dynamic string $X$ and supports minimizer queries in $X$ in $\cO(1)$ time with amortized $\cO(1)$ time per update operation. 
    \item We show that this data structure can be modified to occupy strongly sublinear space without increasing the asymptotic complexity of its operations. To the best of our knowledge, this yields the first algorithm for computing $\Minimizers(S)$ in $\cO(n)$ time \emph{using $\cO(\sqrt{w})$ working space}.
\end{enumerate}

We complement our theoretical results with a concrete application and an experimental evaluation.
\end{abstract}

\newpage

\section{Introduction}
Minimizers is a mechanism for string sampling that has been introduced independently by Schleimer et
al.~\cite{DBLP:conf/sigmod/SchleimerWA03} and by Roberts et al.~\cite{DBLP:journals/bioinformatics/RobertsHHMY04}. Since then, it has been used ubiquitously in modern
sequence analysis applications underlying some of the most widely used tools~\cite{DBLP:journals/bioinformatics/Li16a,DBLP:journals/bioinformatics/Li18,Kraken}.
Minimizers also form the basis of more involved sampling schemes, such as syncmers~\cite{syncmers}, strobemers~\cite{strobemers}, mod-minimizers~\cite{DBLP:conf/wabi/KoerkampP24}, and bd-anchors~\cite{DBLP:journals/tkde/LoukidesPS23}. For more information, see~\cite{DBLP:journals/csur/ChikhiHM21,DBLP:journals/jcb/ZhengMK23}.

We start by providing the standard definition of minimizers.

\begin{definition}[Minimizers]
Let $S = S[0] \ldots S[n-1]$ be a string over an alphabet $\Sigma$, $w \geq 2$ and $k \geq 1$ be two integers, and $\rho=(\Sigma^k,\leq)$ be a total order on $\Sigma^k$. The \emph{minimizer} of the window $X=S[i \dd i + w + k - 2]$ of $S$ is the
smallest position in $[i, i + w - 1]$ where the lexicographically smallest length-$k$ substring of
$S[i \dd i + w + k - 2]$ starts based on $\rho$. The set of minimizers for all $i\in[0,n-w-k+1]$ is the set $\Minimizers(S)$ of the minimizers of $S$.   
\end{definition}

Two standard instantiations of the $\rho$ order are the lexicographic order (e.g., when $\Sigma$ is a totally ordered alphabet~\cite{DBLP:conf/cpm/0001ALP24}) and the total order implied by the Karp-Rabin (KR) rolling hash of the elements of $\Sigma^k$~\cite{DBLP:journals/ibmrd/KarpR87}. For efficiency purposes, practitioners mostly use the latter order. We will henceforth assume the KR order for efficiency unless explicitly stated otherwise. The following proposition is well-known about computing the set $\Minimizers(S)$.

\begin{proposition}[e.g., see~\cite{DBLP:journals/tkde/LoukidesPS23}]\label{the:minimizers}
For any string $S$ of length $n$ over an alphabet $\Sigma$, two integers $w\ge 2$ and $k\ge 1$, and an order $\rho=(\Sigma^k,\leq)$, $\Minimizers(S)$ can be computed in $\cO(n)$ time.    
\end{proposition}

In particular, if we can compute the KR fingerprint (KRF) of a length-$k$ substring of $S$ in $\cO(1)$ amortized time, we can apply a folklore algorithm\footnote{For example, see the description at \url{https://codeforces.com/blog/entry/71687}.}, which computes the minimum elements in a sliding window of
size $w$ in $\cO(1)$ amortized time per window using $\cO(w)$ working space (space on top of $S$). 
It is thus natural to pose the following two questions:
\begin{center}
\textbf{Question 1}: \emph{Can we efficiently support other dynamic updates on the window?}
\end{center}
\begin{center}
\textbf{Question 2}: \emph{Can we improve on the $\cO(w)$ working space}?
\end{center}

\subparagraph{Problem Formulation}

Let $X$ be a string of length $\ell$ over alphabet $\Sigma$. 
The string $X$ is, in general, called \emph{dynamic} in the literature~\cite{DBLP:conf/soda/GuFB94,DBLP:journals/talg/AmirLLS07,DBLP:conf/cpm/CliffordGK0U22} when it can be modified through three different types of \emph{edit} operations: letter insertion, deletion, or substitution. Here, we consider a \emph{semi-dynamic} model, where one is allowed to insert or delete a letter at any of its ends.

More formally, we define the four following types of modifications that are allowed on $X$:
\begin{itemize}
\item \AL($a$) that updates the string $X$ into $aX$, for $a\in\Sigma$;
\item \AR($a$) that updates the string $X$ into $Xa$, for $a\in\Sigma$;
\item \RL() that updates the string $X=aX'$ into $X'$, for $a\in\Sigma$;
\item \RR() that updates the string $X=X'' a$ into $X''$, for $a\in\Sigma$.
\end{itemize}

We will collectively call the above four types of modifications \emph{border modifications} and such a string $X$ a \emph{semi-dynamic} string. A \emph{semi-dynamic minimizer} data structure, for a fixed integer $k\geq 1$ and a semi-dynamic string $X$, is then a data structure that maintains $X$ to support queries asking for the minimizer of $X$; a minimizer query must return the \emph{position}
(and the KRF \emph{value} if needed) of the corresponding substring.
Such a data structure should be much more efficient than the naive comparison of all its length-$k$ substrings.
Note that border modifications form a subset of the most commonly used edits, and as such, problems concerning those can be answered using more general solutions (see \cref{sec:general-queue}). 

\subparagraph{Our Results and Roadmap}

Interestingly, it turns out that while the problem of a \emph{dynamic minimizer} data structure (i.e., the problem for a \emph{dynamic} string $X$) does not seem solvable faster than $\Oh(\log \ell)$ time per edit, in the case of border modifications, a faster solution is indeed possible.
At the same time, a semi-dynamic minimizer data structure is already quite powerful, as it generalizes the most commonly used sliding-window approach (cf.~\cref{the:minimizers}) and the computation of minimizers on a trie~\cite{DBLP:conf/icde/Gabory0LPZ24}, for which a standard  approach would use \cref{the:minimizers} on every leaf-to-root path.
In particular, we make the following contributions:

\begin{enumerate}
    \item In \cref{sec:prel}, we show a warm-up dynamic minimizer data structure that supports $\cO(1)$-time minimizer queries with: (i) $\Oh(\log \ell)$ time per border modification; and (ii) $\Oh(k\log \ell)$ time per general edit. We also provide evidence as to why any improvement over the cost for general edits is unlikely (and, in fact, impossible in the comparison model). 
    \item In \cref{sec:DS}, we show a semi-dynamic minimizer data structure that supports $\cO(1)$-time minimizer queries with amortized $\cO(1)$ time per update operation.
    \item In \cref{sec:SE}, we show that the above semi-dynamic minimizer data structure can be modified to occupy strongly sublinear space without increasing the asymptotic complexity of its operations. To the best of our knowledge, this yields the first algorithm for computing $\Minimizers(S)$ in $\cO(n)$ time using \emph{using $\cO(\sqrt{w})$ working space}.
    \item In \cref{sec:exp}, we present a concrete application of our data structure from \cref{sec:DS} for computing the set of minimizers for all the length-$\ell$ \emph{paths} of a trie~\cite{DBLP:conf/icde/Gabory0LPZ24}. We complement this with an experimental evaluation that demonstrates the benefits of our approach.
\end{enumerate}

\subparagraph{Related Work} The sliding-window and the semi-dynamic models have been considered for maintaining the suffix tree~\cite{DBLP:journals/algorithms/BrodnikJ18,DBLP:journals/cacm/FialaG89,DBLP:conf/dcc/Larsson96}, the directed acyclic word graph~\cite{DBLP:journals/jda/InenagaSTA04,DBLP:conf/spire/SenftD08}, and for numerous other string-processing tasks~\cite{DBLP:journals/tcs/AkagiKMNIBT22,DBLP:journals/iandc/CrochemoreHKMPR20,DBLP:journals/algorithmica/MienoFNIBT22,DBLP:conf/iwoca/MienoF22,DBLP:journals/ipl/MienoWNIBT22,DBLP:conf/cpm/BannaiCR24}. In fact, the first works on dynamic variants of the longest common subsequence and edit distance problems considered updates that were a subset of prepend, append, and delete the first or last letter~\cite{DBLP:journals/jacm/Thorup07,DBLP:journals/siamcomp/LandauMS98,DBLP:journals/jda/KimP04,DBLP:conf/fct/IshidaIST05}.

\section{Preliminaries}\label{sec:prel}

An \emph{alphabet} $\Sigma$ is a finite set of elements called \emph{letters}.
A sequence $X[0]X[1] \ldots X[n-1]$ of letters from $\Sigma$ is called a \emph{string} $X$ of length $n=|X|$.
By $X[i \dd j] = X[i]X[i+1] \ldots X[j]$, we denote a \emph{fragment} of $X$. Strings that are fragments of $X$ are called \emph{substrings} of $X$.

In \cref{sec:string-comp}, we detail how the substrings of $X$ can be efficiently compared.
In \cref{sec:occs}, we detail how the starting positions of the substrings of $X$ are maintained.
In \cref{sec:general-queue}, we present a simple
dynamic minimizer data structure and explain why it is unlikely to be improved.
In \cref{sec: Sliding window}, we present the sliding-window approach for computing $\Minimizers(S)$.

\subsection{Substring Comparison}\label{sec:string-comp}

The cost of the modifications of the data structures presented in the paper is determined by the number of length-$k$ string comparisons.

The techniques used in the solutions presented in the paper do not impose restrictions on the alphabet $\Sigma$, 
other than the fact that we are given an order $\rho=(\Sigma^k,\leq)$. 

A naive comparison of length-$k$ strings imposes a multiplicative $\Oh(k)$ factor in the time complexities. To avoid this problem, we define $\rho=(\Sigma^k,\leq)$ through the corresponding KRF values. In our setting, the KRF of the new, removed or modified length-$k$ fragment can be computed in $\Oh(1)$ time by modifying the stored KRF of its neighbouring length-$k$ fragment (occurring one position earlier or later).

Throughout the paper we store the handle to substring comparison (either a pointer to its position in the string or the KRF) as a \emph{value} which occupies $\Oh(1)$ space and can be computed in $\Oh(1)$ time. Henceforth, in the complexity analysis of the algorithms, we assume that the $\Oh(1)$ time comparison of the KRF is used.

\subsection{Starting Positions of the Substrings}\label{sec:occs}

By modifying $X$ with $\AL(a)$ or $\RL()$ operations, we shift the starting positions (relative to the beginning of the string) of all substrings without actually altering them in any way.
Since data structures store information about the positioning of substrings and updating all of those would be costly, we always refer to the starting position relative to the staring position of the string in its initial state (in particular, the starting position of a substring can be negative). The relative position of the smallest substring (the minimizer) is obtained by subtracting from the answer the position of the first letter of $X$ (which is stored along the structure and updated in $\Oh(1)$ time) for border modifications or by more complicated data structures for edits. Henceforth, we assume that this approach is shared by all our structures and always refer to this immutable starting position of a substring.

\subsection{General Priority-Queue Approach}\label{sec:general-queue}

Upon inserting a letter to string $X$, we produce a new length-$k$ fragment (when $|X|\ge k$), and upon asking the minimizer query, we want to find the length-$k$ substring with the smallest value -- this alone suggests that the dynamic minimizer data structure is some sort of priority queue storing pairs $(\emph{position}, \emph{value})$ sorted first by value and then by position in the case of substring ties: two or more fragments of $X$ can correspond to one substring of $X$.

Standard implementations of priority queues~\cite{DBLP:books/mg/CormenLRS01}, such as a \emph{heap} data structure or \texttt{C++} \texttt{std::set}, allow efficient \textsf{insert}, \textsf{delete}, and \textsf{find-min} operations. In such data structures, the two modification operations are performed in $\Oh(\log \ell)$ time (where $w=\ell - k +1\leq \ell$ is the number of elements stored in the structure); the smallest element is found in $\Oh(1)$ time.

This immediately gives us the following straightforward result.

\begin{fact}
The general priority-queue implementation of a dynamic minimizer data structure for a fixed integer $k\geq 1$ and a dynamic string $X$ of length $\ell$ supports $\Oh(1)$-time minimizer queries and performs border modifications in $\Oh(\log \ell)$ time and edits in $\Oh(k\log \ell)$ time.
\end{fact}
\begin{proof}
Recall that at any moment we store at most $w=\ell - k +1\leq \ell$ elements.
Upon inserting (resp.~deleting) a new length-$k$ fragment, we can compute its value in $\Oh(1)$ time, and thus we know exactly what $(\emph{position}, \emph{value})$ pairs should be inserted (resp.~deleted). 

For border modifications, exactly one such pair is inserted (resp.~deleted), hence such an operation can be performed in $\Oh(\log \ell)$ time in a heap implementation of priority queue.

The case of general edits is a little more complicated -- a single letter change can affect up to $k$ substrings of $X$, hence the modification of the structure requires deleting and inserting up to $k$ pairs, which gives us $\Oh(k\log \ell)$ time per modification.
Additionally, insertions or deletions in the middle of $X$ change the positioning of the substrings in a non-consistent way. However, this problem can be solved with a data structure (e.g., a segment tree) that stores these position shifts and supports $\Oh(\log \ell)$ time interval updates~\cite{DBLP:books/lib/BergCKO08}.

In both cases, finding the minimizer can be performed by simply asking for the smallest pair (possibly stored along the structure).
\end{proof}

As mentioned in the proof, dealing with the internal changes is more involved and costly, in particular in some settings the $\Oh(\log \ell)$ time per update cannot be avoided (even if we do not need to care for the uneven positioning), as illustrated by the following lemma.

\begin{lemma}
In the comparison model, any dynamic minimizer data structure, for a fixed integer $k\geq 1$ and a semi-dynamic string $X$ of length $\ell$ over a general totally ordered alphabet, cannot support substitution modifications in $o(\log \ell)$ amortized time. 
\end{lemma}
\begin{proof}
We restrict our focus to the case of $k=1$ (for larger $k$ the first letter is the most important anyway) and naive string comparison, which is equivalent to KRF for $k=1$.

Given a sequence of $\ell$ numbers, we can construct a string representing this sequence: the $i$th letter of $X$ is the $i$th number in the sequence.

With the dynamic minimizer structure for substitution operations (that is, a subset of edit operations), we can repeat $\ell$ times a procedure of finding the position of the smallest length-$1$ substring (the number) and replacing the number on that position with the maximum of the whole sequence. Such a procedure sorts the sequence of numbers using $\Oh(\ell)$ 
modifications of the sequence, hence those operations cannot cost less than $\Omega(\log \ell)$ comparisons~\cite{DBLP:books/mg/CormenLRS01}.
\end{proof}

For a constant-size alphabet $\Sigma$ and small values of $k$, one can resort to bucketing (for $k=1$ one can simply store the queue as a set of $|\Sigma|$ buckets); this approach is generally less practical, however, and does not solve additional problems such as supporting positional shifts. In the remainder of the paper, we restrict the operations on $X$ to border modifications in order to design a data structure that is \emph{both theoretically efficient and fast in practice}.

\subsection{Sliding-Window Approach}\label{sec: Sliding window}

Minimizers are most commonly used for string sampling -- a set of string \emph{anchors} is defined as the set of minimizers of all length-$\ell$ fragments of the input string $S$ for some fixed $k$. In this setting, the commonly used method to compute minimizers is the sliding-window approach - the length-$\ell$ fragment starting one position later differs by very little, and hence the update of the variables kept by the algorithms is usually cheap.
From the perspective of minimizers, usually the position of the minimizer stays the same; it can change in two cases: either the new length-$k$ fragment becomes the smallest one; or the current minimizer falls out of the window under consideration.
The standard technique exploits this setting of very restricted changes to store only information about the substrings that can become minimizers in the future.
Formally, the set of the substrings kept consists of only the pairs $(p,v)$ such that there exists no length-$k$ substring in the current window with a representing pair $(p',v')$ such that $p'>p$ and $v'<v$ (the pair is not \emph{dominated} by any other pair).
Thus, the deque (in this case a bidirectional linked-list is enough) 
storing those substrings is sorted non-decreasingly forming a Manhattan skyline (although reversed as we look for the smallest values instead of the largest ones), which allows updates in $\Oh(1)$ amortized time (more details are given in the next section, where the approach is generalised); see \cref{fig: Sliding window}.

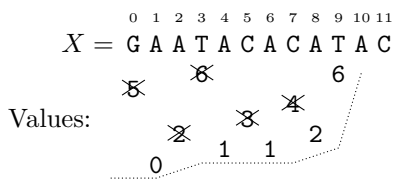
\begin{figure}[ht]
\vspace*{-.4cm}
\centering
\begin{tikzpicture} 
  
\begin{scope}[yshift=0cm, xshift=0cm]
  \foreach \ii [count=\i from 0] in {0,1,2,3,4,5,6,7,8,9,10,11} {
    \node at (\i * 0.3 + 0.3, 0.4) [above] {\tiny \ii};
  }
  \node at (-0.3, 0) [above] {$X=$};
  \foreach \l [count=\i from 0] in {G,A,A,T,A,C,A,C,A,T,A,C} {
    \node at (\i * 0.3 + 0.3, 0) [above] {\tt \texttt{\l}};
  }
  \node at (-0.8, -1) [above] {Values: };
    \foreach \hash [count=\i from 0] in {5,0,2,6,1,3,1,4,2,6} {
    \node at (\i * 0.3 + 0.3, -1.6 + \hash * 0.2) [above] {\tt \hash};
  }
  \foreach \x/\y in {0/5, 2/2, 3/6, 5/3, 7/4} {
    \draw[thin] (0.3 * \x +0.15, -1.45 + \y * 0.2) -- (0.3*\x+0.45, -1.25 + \y * 0.2);
    \draw[thin] (0.3 * \x +0.45, -1.45 + \y * 0.2) -- (0.3*\x+0.15, -1.25 + \y * 0.2); 
  }
  \draw[densely dotted] (0, -1.55) -- (0.6, -1.55) -- (1.2, -1.35) -- (1.8, -1.35) -- (2.4, -1.35) -- (3, -1.15) -- (3.3, -0.15);
  \end{scope}
\end{tikzpicture}
\caption{
Window $X$ (slid over string $S$) with the values for each length-$3$ fragment (for legibility matching their lexicographic order) starting at each position of $X$. 
Values for which there exists a smaller value to the right are crossed out -- they cannot become a minimizer, as this smaller value will leave the window later.
The remaining values form a non-decreasing sequence; notice that, in particular, the last value cannot be crossed out.
}\label{fig: Sliding window}
\end{figure}

\section{Constant-Time Update Data Structure}\label{sec:DS}

In \cref{sec:dequeue}, we show how the folklore sliding-window approach for computing minimizers (see \cref{sec: Sliding window}) can be generalized to support \emph{a subset} of the border modifications.
In \cref{sec:2Stack}, we present our data structure for the \emph{complete} set of border modifications.

We implement the semi-dynamic string $X$ using a standard bidirectional linked list.

\subsection{Generalizing the Sliding-Window Approach}\label{sec:dequeue}

We show that the sliding-window approach with storing information about only the substrings that can become minimizers in the future can be applied directly to obtain a minimizer data structure with a \emph{limited set of border modifications}: $\AL(),\AR(),\RL()$. This then becomes a building block of our data structure for all border modifications. 

In our deque data structure (or in a bidirectional linked list), we store the $(\emph{position}, \emph{value})$ pairs that represent the length-$k$ fragments that can become minimizers of $X$ after some modifications. A pair $(p,v)$ is dominated by a pair $(p',v')$ if $p'<p$ and $v'>v$.
We keep the pairs that are not dominated by any other one (from the representation of all the length-$k$ fragments of $X$) in a deque sorted increasingly by their first element (position), and because of this domination property at the same time sorted non-decreasingly by their second element (value); see \cref{fig: Sliding window}.
We call this the \emph{limited semi-dynamic minimizer} data structure.

\begin{lemma}\label{lem:dequeue}
The limited semi-dynamic minimizer data structure supports $\Oh(1)$-time minimizer queries and performs limited border modifications in $\Oh(1)$ amortized time.
\end{lemma}
\begin{proof}
Let $p$ be the position of the minimizer, $v$ its value, and let $(p_1,v_1)$ be the first element of the deque (retrieved trivially in $\Oh(1)$ time); by the definition of the minimizer $v_1\ge v$.
If $p_1>p$, then the pair $(p,v)$ would appear in the deque before $(p_1,v_1)$ -- it is clearly not dominated by any pair since it has the smallest value.
If $p_1<p$ and $v_1>v$, then the pair $(p_1,v_1)$ would be dominated by the pair $(p,v)$ and therefore would not appear in the dequeue. This shows that $p_1 = p$ is the position of the minimizer and that it can be retrieved in $\Oh(1)$ time.

When updating the structure we have to consider the three possible modifications:
\begin{itemize}
\item $\AR(a)$: we insert a new pair $(p, v)$ representing the new fragment to the structure; as one with the largest position it will have to be stored in the deque (as its last element). We also need to remove all pairs stored in the deque with a value larger than $v$, which is done in time linear in the number of removals as the pairs are stored sorted.
\item $\RL()$: we remove a single fragment from the structure. If its position is equal to $p_1$, then the pair $(p_1,v_1)$ is removed -- the second element stored in the deque represents the new minimizer, as every fragment in between has strictly larger value.
\item $\AL(a)$: we only want to add the new pair $(p, v)$ if it is not dominated by any pair currently stored in the structure. To check this property, it is enough to compare value $v$ with $v_1$: if $v\le v_1$, then $(p, v)$ should be added to the structure and represent the new minimizer (other stored pairs cannot dominate this pair as their values can only be greater or equal to $v_1$); otherwise, no modification to the structure is required.  
\end{itemize}

The cost of operation $\AL()$ (resp.~$\RL()$) is $\Oh(1)$: a single comparison and single insertion (resp.~deletion) of a pair. The cost of $\AR()$ is amortized by the number of elements removed from the dequeue; namely, by the number of insertions and, in turn, the number of $\AL()$ and $\AR()$ operations performed earlier.
\end{proof}

By symmetry, the technique underlying \cref{lem:dequeue} can be adapted to support operations $\AR(), \AL()$, and $\RR()$, and with some modifications to solve the problem of a window sliding both sides over string $S$. Unfortunately, however, this technique does not extend to the full set of border modifications because of the following fact.

\begin{fact}\label{fct:Double Removals}
When both $\RL()$ and $\RR()$ operations are allowed on string $X$ every length-$k$ fragment of $X$ can become the minimizer in the future as one can perform subsequent $\RL()$ or $\RR()$ operations so that $X$ contains only this single length-$k$ fragment.
\end{fact}

This shows that keeping the information about all the fragments that can become minimizers in the future brings us back to the general approach with storing information about all length-$k$ fragments; in the next section, we show how to work around this problem.

\subsection{The Semi-Dynamic Minimizer Data Structure}\label{sec:2Stack}

Our approach consists of relaxing the property of storing the $(\emph{position},\emph{value})$ pairs for all fragments that can become minimizers, to storing them only for the fragments that can become minimizers in the near future. When this ``time period'' expires (when we are no longer sure that the current minimizer is stored), we rebuild the whole structure from scratch.


Recall that $X$ is our semi-dynamic string of length $\ell$.
We distinguish its single position $x$ (preferably close to the center of the string). We partition the set of length-$k$ fragments of $X$ into those that start before position $x$ and those that start at or after position $x$. For each of those subsets, we store a separate queue structure represented by a stack.

The \emph{left stack} represents the pairs $(p,v)$ such that there exists no pair $(p',v')$ with $p'>p$ and $v'<v$ among the left subset of length-$k$ fragments -- that is exactly the same as in \cref{sec:dequeue}. This time, a stack (or a one-directional linked list) is enough with $(p_1,v_1)$ stored on top. Symmetrically, the \emph{right stack} represents the pairs $(p,v)$ such that there exists no pair $(p',v')$ with $p'<p$ and $v'\le v$ among the right subset of length-$k$ fragments.
This time, the represented subsequence is strictly increasing, as we do not need to store multiple pairs with the same value (only the one with the leftmost position is important).
Again, pairs are stored in the stack in the order of their positions - this time, the one with the rightmost position (hence again with the smallest value) is at the top. See \cref{fig:2Stack}.

Recall that for efficiency, the values are computed as KRFs.
Thus, we also store the KRF values of the leftmost and rightmost length-$k$ fragments of $X$ to allow for fast computation of the KRF of the length-$k$ fragment that is inserted or deleted during the modification of $X$.

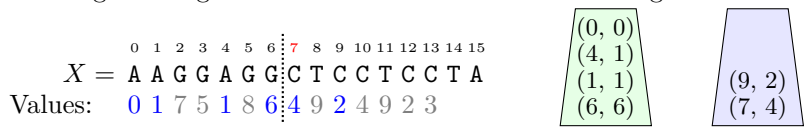
\begin{figure}[ht]
\vspace*{-.4cm}
\centering
\begin{tikzpicture} 
  
\begin{scope}[yshift=0cm, xshift=0cm]
  \foreach \ii/\c [count=\i from 0] in {0/black,1/black,2/black,3/black,4/black,5/black,6/black,7/red,8/black,9/black,10/black,11/black,12/black,13/black,14/black,15/black} {
    \node at (\i * 0.3 + 0.3, 0.4) [above, \c] {\tiny \ii};
  }
  \node at (-0.3, 0) [above] {$X=$};
  \foreach \l [count=\i from 0] in {A,A,G,G,A,G,G,C,T,C,C,T,C,C,T,A} {
    \node at (\i * 0.3 + 0.3, 0) [above] {\tt \texttt{\l}};
  }
  \node at (-0.8, -0.4) [above] {Values: };
    \foreach \hash/\c [count=\i from 0] in {$0$/blue,$1$/blue,$7$/gray,$5$/gray,$1$/blue,$8$/gray,$6$/blue,$4$/blue,$9$/gray,$2$/blue,$4$/gray,$9$/gray,$2$/gray,$3$/gray} {
    \node at (\i * 0.3 + 0.3, -0.4) [above, \c] {\tt \hash};
  }
  \draw[densely dotted,thick] (2.25,-0.4) -- (2.25, 0.8);
  \end{scope}
  \begin{scope}[yshift=0cm, xshift=6cm]
    \begin{scope}[yshift=0cm, xshift=0cm]
    \draw [fill=green!10!white] (-0.1,-0.45) -- (1.1,-0.45) -- (0.9, 1.1) -- (0.1, 1.1)-- cycle;
     \foreach \pos/\hash [count=\i from 0] in {6/6,1/1,4/1,0/0} {
        \node at (0.5, -0.5 + \i * 0.35) [above] {\small (\pos, \hash)};
        }
    \end{scope}
    \begin{scope}[yshift=0cm, xshift=2cm]
    \draw [fill=blue!10!white] (-0.1,-0.45) -- (1.1,-0.45) -- (0.9, 1.1) -- (0.1, 1.1)-- cycle;
     \foreach \pos/\hash [count=\i from 0] in {7/4,9/2} {
        \node at (0.5, -0.5 + \i * 0.35) [above] {\small (\pos, \hash)};
        }
    \end{scope}
  \end{scope}
\end{tikzpicture}
\caption{
A graphical representation of the data structure. String $X$ with the values for each length-$3$ fragment (for simplicity matching their lexicographic order) starting at each position of $X$. Position $x$ drawn in red divides $X$ into two parts -- for each one, the ``interesting'' position-values are drawn in blue; the ones in gray are not represented but they may be computed again during the full structure rebuild. One can notice that the distinguished values can repeat in the left part but not in the right one. On the right, the ``interesting'' $(\emph{position},\emph{value})$ pairs are stored in two stacks. Note that, in particular, the positions $x$ and $x-1$ are always represented (the stacks are nonempty).
}\label{fig:2Stack}
\end{figure}

We call this the \emph{two-stack semi-dynamic minimizer} data structure.
Each of the stacks stores as its top element the only candidate for the minimizer over its subset of length-$k$ fragments; hence, the minimizer can be obtained using a single comparison (or stored along the structure to support $\Oh(1)$-time queries in case of brute-force comparison).
The modifications of the left stack are performed in the same way as in \cref{sec:dequeue} without the $\AR$ operation, which is the only one whose cost was amortized (we never remove multiple pairs in a single operation). By symmetry, the right stack is handled analogously (with the simple change of handling the ties of values upon insertion to not store two pairs with the same value).

However, this time, a single exceptional case can happen: if we try to remove a pair from an empty stack, the behaviour is undefined; we would like to remove something from the other stack, but as mentioned in \cref{fct:Double Removals} this causes problems.
In such a situation, we restore the ``basic state of'' the structure by ``dividing the single stack into two stacks'' as follows:
\begin{enumerate}
    \item We recompute all the discarded pairs $(\emph{position},\emph{value})$;\footnote{We do not compute all of the pairs at once, instead we simply perform $\ell$ insertions on an empty structure. This is important for \cref{sec:SE} as such approach requires only $\Oh(1)$ extra construction space.}
    \item We set the value of $x$ to the starting position of $X$ plus $\lfloor (\ell-k)/2\rfloor$;
    \item We divide the set of length-$k$ fragments accordingly;
    \item We compute the left and right stacks.
\end{enumerate}

The above four steps can be clearly implemented in $\Oh(\ell)$ time.\footnote{This approach of rebuilding the whole structure from scratch is often used in practice to achieve amortized $\Oh(1)$-time per operation (e.g., by data structures like the \texttt{C++ std::vector} when the structure gets too large to fit in the allocated memory).}
Note that the cost of rebuilding is amortized by the number of insertions plus the number of deletions performed since the last rebuild.
Indeed, assume that after the last rebuild both stacks had size $h$, and now one stack is of size $0$, while the other one is of size $\ell$. This means that we had to perform at least $h$ operations $\RL()$ and at least $\ell-h$ operations $\AR()$.

The correctness of the approach in the standard case follows from \cref{lem:dequeue}. In the special case, we rebuild everything from scratch; hence we cannot miss any $(\emph{position},\emph{value})$ pair needed. Upon recomputing the two stacks, we fall back to the standard case: all values that may be needed before the next rebuild are stored. 
We have arrived at the main result.

\begin{theorem}\label{thm: 2Stack}
The two-stack semi-dynamic minimizer data structure supports $\Oh(1)$-time minimizer queries and performs border modifications in $\Oh(1)$ amortized time.
\end{theorem}

\section{Space-Efficient Data Structure}\label{sec:SE}

In the previous section, we showed how to obtain a fast data structure by keeping only the pairs that may become important in the near future. Still, those pairs translate to keeping the representation of all the $\ell-k+1$ length-$k$ fragments of $X$ if the sequences of their values before and after position $x$ are monotonous. In this section, we show how to extend the two main ideas of our data structure -- keeping only the pairs that may become useful in the near future and re-computation -- to further reduce the space occupied by the data structure.

We implement the semi-dynamic string $X$ using a standard bidirectional linked list.

\subsection{High-Level Idea: Two-Layer Data Structure}

The high-level idea consists of partitioning our semi-dynamic string $X$ into blocks; each block is represented by a pair $(\emph{position},\emph{value})$ for the length-$k$ fragment with the smallest value among the ones starting within this block. We then treat the blocks as if those were just single length-$k$ fragments and store information about those in stacks.

Additionally, we treat the first block of $X$ as if it were the whole left part of $X$ (all the length-$k$ fragments starting before position $x$), and represent all of its length-$k$ fragments that may correspond to minimizers in the future (before the stack is emptied or the whole structure is rebuilt) just as in the previous section.
We analogously treat the last block of $X$ as if it were the whole right part of $X$ (all length-$k$ fragments starting at or after $x$).

The total number of pairs stored is bounded by the number of blocks plus the total size of a constant number of blocks.
Conceptually, we have two layers: the \emph{first layer} corresponds to the length-$k$ fragments of $X$; and the \emph{second layer} corresponds to the blocks of $X$.

\subsubsection{Fixed-Length Blocks}

We start with the simplest block division design, where each block has the same fixed length, to show the technical details of the approach in a possibly clear setting.

Let $c>1$ be a fixed integer and let $B_\kappa=\{x+\kappa\cdot c ,x+\kappa\cdot c + 1, \cdots, x - 1 + (\kappa+1)\cdot c\}$, for some integer $\kappa$. Recall that $x$ is a position on the string $X$.
This divides the length-$k$ fragments of $X$ into blocks of length $c$; i.e., length-$k$ fragments starting at positions in $B_\kappa$ form a \emph{block}.

Say that we want to build the structure for a nonempty string $X$; for example this happens when the full structure is rebuilt. 
For the first and last nonempty blocks (containing some length-$k$ fragment), we build stacks in the same way as in \cref{sec:2Stack}. Additionally, for each ``internal'' block (if any) we find the fragment with the smallest value (breaking ties for the one with the smallest position) and build two stacks for only those fragments, again just as in \cref{sec:2Stack}.
The blocks $B_\kappa$ with $\kappa<0$ are represented with the left second-layer stack, and the blocks $B_\kappa$ with $\kappa\ge0$  are represented with the right second-layer stack; see \cref{fig: 2SLayer}.
We call this the \emph{two-layer semi-dynamic minimizer} data structure. 

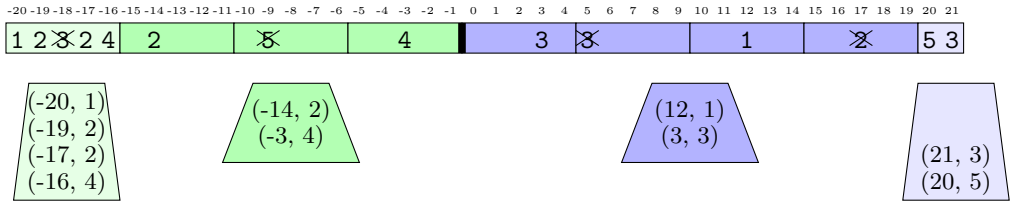
\begin{figure}[ht]
\vspace*{-.4cm}
\centering
\begin{tikzpicture} 
\begin{scope}[yshift=0cm, xshift=0cm]
  \draw [fill=green!10!white] (0,0) rectangle (1.5, 0.4);
  \draw [fill=green!30!white] (1.5,0) rectangle (3, 0.4);
  \draw [fill=green!30!white] (3,0) rectangle (4.5, 0.4);
  \draw [fill=green!30!white] (4.5,0) rectangle (6, 0.4);
  \draw [fill=blue!30!white] (6,0) rectangle (7.5, 0.4);
  \draw [fill=blue!30!white] (7.5,0) rectangle (9, 0.4);
  \draw [fill=blue!30!white] (9,0) rectangle (10.5, 0.4);
  \draw [fill=blue!30!white] (10.5,0) rectangle (12, 0.4);
  \draw [fill=blue!10!white] (12,0) rectangle (12.6, 0.4);
  
  \foreach \i in {-20,-19,-18,...,21} {
    \node at (\i * 0.3 + 6.15, 0.4) [above] {\fontsize{4pt}{4pt}\selectfont \i};
  }
  
  \foreach \pos/\hash in {-20/1,-19/2,-18/3,-17/2,-16/4,-14/2,-9/5,-3/4,3/3,5/3,12/1,17/2,20/5,21/3} {
    \node at (\pos * 0.3 + 6.15, -0.05) [above] {\tt \hash};
  }
  \foreach \x in {-18, -9, 5, 17} {
    \draw[thin] (0.3 * \x + 6, 0.1) -- (0.3*\x + 6.3, 0.3);
    \draw[thin] (0.3 * \x + 6, 0.3) -- (0.3*\x + 6.3, 0.1);
  }

  \draw [line width = 1 mm] (6,0) -- (6,0.4);
  \end{scope}
  \begin{scope}[yshift=-1.5cm, xshift=0cm]
    \begin{scope}[yshift=0cm, xshift=0.3cm]
    \draw [fill=green!10!white] (-0.2,-0.45) -- (1.2,-0.45) -- (1, 1.1) -- (0, 1.1)-- cycle;
     \foreach \pos/\hash [count=\i from 0] in {-16/4,-17/2,-19/2,-20/1} {
        \node at (0.5, -0.5 + \i * 0.35) [above] {\small (\pos, \hash)};
        }
    \end{scope}
    \begin{scope}[yshift=0cm, xshift=3.25cm]
    \draw [fill=green!30!white] (-0.4,0.05) -- (1.4,0.05) -- (1, 1.1) -- (0, 1.1)-- cycle;
     \foreach \pos/\hash [count=\i from 0] in {-3/4,-14/2} {
        \node at (0.5, 0.1 + \i * 0.35) [above] {\small (\pos, \hash)};
        }
    \end{scope}
    \begin{scope}[yshift=0cm, xshift=8.5cm]
    \draw [fill=blue!30!white] (-0.4,0.05) -- (1.4,0.05) -- (1, 1.1) -- (0, 1.1)-- cycle;
     \foreach \pos/\hash [count=\i from 0] in {3/3,12/1} {
        \node at (0.5, 0.1 + \i * 0.35) [above] {\small (\pos, \hash)};
        }
    \end{scope}
    \begin{scope}[yshift=0cm, xshift=12cm]
    \draw [fill=blue!10!white] (-0.2,-0.45) -- (1.2,-0.45) -- (1, 1.1) -- (0, 1.1)-- cycle;
     \foreach \pos/\hash [count=\i from 0] in {20/5,21/3} {
        \node at (0.5, -0.5 + \i * 0.35) [above] {\small (\pos, \hash)};
        }
    \end{scope}
  \end{scope}

\end{tikzpicture}
\caption{
A graphical representation of the two-layer data structure for $c=5, x=0$.
The ``internal'' blocks (with darker colour on the figure) are represented by one fragment each (the one with the smallest value).
We store one second-layer stack representing the internal blocks on each side, and one first-layer stack for each border block.
Note that in the process of updating the structure, there are either one or two such blocks for each side of the string.
}\label{fig: 2SLayer}
\end{figure}

\begin{lemma}\label{lem: 2 Layer Fixed Length}
Suppose that we have read-only access to a semi-dynamic string $X$ of length $\ell$. The two-layer semi-dynamic minimizer data structure with blocks of length $c$, for any integer $c>1$, supports $\Oh(1)$-time minimizer queries, performs border modifications in $\Oh(1)$ amortized time, and occupies $\Oh(c+\frac{\ell}{c})$ space.
\end{lemma}
\begin{proof}
The minimizer is stored as the top element of one of the stacks.
If it comes from a ``border'' block, then the correctness follows from \cref{thm: 2Stack}.
Similarly, if the minimizer belongs to one of the ``internal'' blocks, then it must for sure be its representing element, thus again it is stored as the top element of the second layer stack by the proof of \cref{thm: 2Stack}.

When the structure is updated, and adding a new fragment does not require starting a new block (respectively, we do not attempt to remove an element from an empty first-layer stack), then the operation clearly takes $\Oh(1)$ time like in the previous section.

If a new block gets started upon the insertion of a letter, we simply create an empty stack for the new block and insert the pair representing the new fragment there. This means that we can have more than one first-layer stacks for this side of $X$. For space reasons if there are three such stacks for a single side of $X$, we remove the most internal one; i.e., insert its top element into the second-layer stack (if it satisfies the conditions of not being dominated) and simply discard the rest of the stack -- this takes $\Oh(c)$ time. Symmetrically, when we want to remove an element from an empty first-layer stack, we simply use the first-layer stack for the next block, and if such a stack is not in the memory, we recompute it from scratch in $\Oh(c)$ time (plus possibly remove the top element of the second-layer stack).

Like in the vanilla version of the two-stack semi-dynamic minimizer structure, if we pass the position $x$, the whole structure gets recomputed from scratch.

The complexity analysis is similar to the one of \cref{sec:2Stack}. The only costly operations are rebuilds.
In the case of a full rebuild, the costs are covered by the number of insertion and deletion operations since the last full rebuild was performed. 

The re-computation of a first-layer stack takes $\Oh(c)$ time, yet, for such an operation to happen at least $c$ letters had to be deleted from this side of $X$ since the previous re-computation or stack removal.
Indeed, whether the previous such operation was a re-computation or a removal one full stack was in the memory afterwards. 
Analogously, in case of removal of a first-layer stack, one had to insert at least $c$ letters on this side of $X$.
In case of the first re-computation or removal after the full rebuild, we charge its cost to this full rebuild (the cost is correctly covered even if $\ell<c$).
In total, this shows that the amortized cost of operations is $\Oh(1)$.

The space usage follows from the fact that there are at most four first-layer stacks at a time, each of size $\Oh(c)$, plus two second-layer stacks, each of size $\Oh(\ell/c)$.
\end{proof}

\subsubsection{\texorpdfstring{$\Oh(\sqrt{\ell})$}{O(√l)}-Space Data Structure}\label{sec: 2layer sqrt(n)}

Note that previously, we never actually used the fact that the blocks are of the same length; the complexity analysis relied only on the property that the more ``internal'' block (for which the stack is removed or recomputed) is up to a constant factor larger (or actually smaller) than the block that is further outside (whose completion or emptying initializes the operation).

Let $B'_\kappa = \{x+\kappa^2+1,x+\kappa^2+2,\cdots, x+(\kappa+1)^2\}$, for $\kappa\ge 0$, and $B'_\kappa=\{x - \kappa^2, x -\kappa^2+1,\cdots, x - (\kappa+1)^2 - 1\}$, for $\kappa<0$. The blocks on each side have lengths $1,3,5,7,\cdots$ from the middle (from position $x$) outwards.
In this way, for $\ell=|X|$ we have at most $\sqrt{\ell}+1$ blocks on each side ($1+3+5+ \cdots +2\kappa+1 = (\kappa+1)^2$, and $\kappa^2<\ell$) of $x$. The size of the outer blocks is also bounded by $\Oh(\sqrt{\ell})$. 
We call this partitioning of $X$ \emph{progressing blocks}, and
from \cref{lem: 2 Layer Fixed Length}, we obtain the following result.

\begin{theorem}
Suppose that we have read-only access to a semi-dynamic string $X$ of length $\ell$. The two-layer semi-dynamic minimizer data structure with progressing blocks supports $\Oh(1)$-time minimizer queries, performs border modifications in $\Oh(1)$ amortized time, and occupies $\Oh(\sqrt{\ell})$ space.\label{the:2layers}
\end{theorem}

\cref{the:2layers} (even if it is more powerful than what is needed) has the following implication.

\begin{corollary}
Suppose that we have read-only access to a string $S$ of length $n$ over an alphabet $\Sigma$, two integers $w\ge 2$ and $k\ge 1$, and an order $\rho=(\Sigma^k,\leq)$. The set $\Minimizers(S)$ can be computed in $\cO(n)$ time using $\cO(\sqrt{w})$ working space. 
\end{corollary} 

\subsection{Multi-Layer Data Structure}

The above design naturally generalizes to multiple layers; blocks can be aggregated into second-layer blocks, just as fragments were grouped into blocks.
In this way for a three-layer semi-dynamic minimizer data structure, the space occupied becomes $\Oh(\sqrt[3]{\ell})$.
When using $h$ layers, the occupied space becomes $\Oh(h\cdot \ell^{1/h})$. At the same time, the cost of operations actually grows - there are at most $4h - 2$ stacks at the same time, hence the cost of finding the minimizers grows to $\Oh(h)$ (although it can be stored along the structure and have its cost covered by the modifications). The cost of border modifications also increases to amortized $\Oh(h)$; a single modification pays for the amortization of every layer.

In the limit, each block of level $h$ consists of a constant number of blocks of length $h-1$, and new layers are introduced when needed, that is, the blocks form a binary balanced tree over the length-$k$ fragments of $X$. 
The space occupied by the data structure becomes $\Oh(\log \ell)$, while operations require $\Oh(\log \ell)$ amortized time.

\section{Application and Experimental Evaluation}\label{sec:exp}

\subsection{Application: Minimizers on a Trie}

A \emph{weighted string} $W$ of length $n$ is a sequence of $n$ probability distributions over an alphabet $\Sigma$~\cite{DBLP:conf/cpm/BartonKPR16,DBLP:journals/iandc/BartonK0PR20,DBLP:journals/jea/Charalampopoulos20}. Such sequences are used to represent uncertain data as well as a succinct model for many similar but not equal strings (e.g., DNA sequences representing a pangenome)~\cite{DBLP:conf/icde/Gabory0LPZ24}.

In such a compact representation, one is interested in finding the positions where a given pattern $P$ occurs with a sufficiently high probability. Formally, given a \emph{weight threshold} $1/z$, we say that string $P$ occurs in $W$ at position $i$ if the product of probabilities of the letters of $P$ at positions $i,\ldots, i+|P|-1$ in $W$ is at least $1/z$~\cite{DBLP:conf/cpm/AmirCIKZ06}. Although indexing methods for standard strings are space-efficient, for weighted ones, the state-of-the-art indexes occupy $\Oh(nz)$ space~\cite{DBLP:conf/cpm/BartonKPR16,DBLP:journals/iandc/BartonK0PR20,DBLP:journals/jea/Charalampopoulos20}, making them impractical~\cite{DBLP:conf/icde/Gabory0LPZ24}.
A space-efficient minimizer-based index for $W$ and patterns of length $|P|\ge \ell$ was shown in~\cite{DBLP:conf/icde/Gabory0LPZ24} with two construction methods.
The first method is based on computing the minimizers in length-$\ell$ windows of a so-called $z$-estimation of $W$ composed of $\lfloor z\rfloor$ standard length-$n$ strings~\cite{DBLP:journals/iandc/BartonK0PR20}.
The second is based on computing minimizers in a trie of size $\Oh(nz)$ representing $W$~\cite{DBLP:conf/cpm/BartonKPR16}: compute the minimizers of every length-$\ell$  path starting in a node and going towards the root; see \cref{fig: Tree}. 

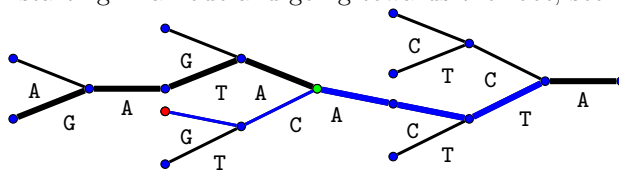
\begin{figure}[ht]
\vspace*{-.4cm}
\centering
\begin{tikzpicture}

\tikzstyle{dot}=[inner sep=0.04cm, circle, draw]

\node[dot,fill=blue] (t) at (0,0) {};
\foreach \name/\dy/\parent/\l/\thick in {
  t1/0/t/\texttt{A}/2,
  t2/0.5/t1/\texttt{T}/2,
  t3/-0.2/t2/\texttt{C}/2,
  t4/-0.2/t3/\texttt{A}/2,
  t5/-0.4/t4/\texttt{A}/2,
  t6/0.4/t5/\texttt{T}/2,
  t7/0/t6/\texttt{A}/2,
  t8/0.4/t7/\texttt{G}/2,
  a1/-0.5/t1/\texttt{C}/1,
  a2/-0.4/a1/\texttt{C}/1,
  a3/0.4/a1/\texttt{T}/1,
  b1/0.5/t2/\texttt{T}/1,
  c1/0.5/t4/\texttt{C}/1,
  c2/-0.2/c1/\texttt{G}/1,
  c3/0.5/c1/\texttt{T}/1,
  d1/-0.4/t5/\texttt{G}/1,
  e1/-0.4/t7/\texttt{A}/1
} {
  \node[dot,fill=blue] (\name) at ($(\parent)-(1,\dy)$) {};
  \draw[black,line width = \thick *  0.4 mm] (\parent)--(\name) node[midway,auto] {\small \l};
}
\foreach \name/\parent/\thick in {c2/c1/1,c1/t4/1,t4/t3/2,t3/t2/2,t2/t1/2} {
\draw[blue,line width = \thick *  0.4 mm] (\parent)--(\name);
}
 \node[dot,fill=red] at (c2) {};
 \node[dot,fill=green] at (t4) {};

\end{tikzpicture}
\caption{
Illustration of the problem of computing minimizers for all the length-$\ell$ paths of a trie (consisting of a \emph{heavy path} and small subtrees hanging out), starting from nodes and going towards the root, exactly as considered in~\cite{DBLP:conf/icde/Gabory0LPZ24}. For $(\ell=5, k=2)$, the smallest (in lexicographic order) length-$2$ fragment of $\texttt{GCACT}$ (the length-$5$ string spelled on the blue path starting at the red node) is $\texttt{AC}$, hence the green node is reported as part of the output of the algorithm from~\cite{DBLP:conf/icde/Gabory0LPZ24}.
}\label{fig: Tree}
\end{figure}

Although the trie itself has size $\Oh(nz)$ and can be constructed in $\Oh(nz)$ time~\cite{DBLP:conf/cpm/BartonKPR16} (just as the $z$-estimation~\cite{DBLP:journals/iandc/BartonK0PR20}), a key advantage of it, compared to the $z$-estimation, is that \emph{it does not need to be constructed explicitly}~\cite{DBLP:conf/icde/Gabory0LPZ24}. 
The trie approach still poses natural limitations; while the total size of the trie is $\Oh(nz)$, it can actually have $\Theta(nz)$ leaves. 
Hence, using the standard sliding-window approach for those leaves \emph{only} would take $\Oh(nz\ell)$ time; and applied directly to every string represented by the trie would take $\Oh(n^2 z)$ time. Thus the algorithm presented in~\cite{DBLP:conf/icde/Gabory0LPZ24} simulates the trie in DFS order, computes the minimizers using a structure like the one from \cref{sec:general-queue}, and thus \emph{requires only $\Oh(n)$ construction space} (because $n$ is the maximal string depth in the trie) and $\Oh(nz\log \ell)$ time. 
Using our \emph{semi-dynamic minimizer} data structure this improves to $\Oh(n)$ space and $\Oh(nz)$ time.
Note that alternatively, one could use a data structure for answering range minimum queries on a tree (cf.~\cite{DBLP:journals/algorithmica/DemaineLW14}); however, this translates to preprocessing the whole trie, which would take $\Omega(nz)$ space.

In the following section, we evaluate the practical runtime of the methods described in the paper to compute minimizers in the above two settings: (i) computing minimizers in a collection of standard strings ($z$-estimation); and (ii) computing minimizers on a trie.

\subsection{Experimental Evaluation}

While the multi-layer data structure from \cref{sec:SE} has nice theoretical properties, still, in practice, the size of the one-layer structure is already sublinear: assuming that values of fragments are drawn at random from a large enough universe the number of non-dominated pairs is $\Theta(\log \ell)$ on average.
At the same time, the vanilla version of the structure (\cref{sec:2Stack}) is much faster than the multi-layer one, and we have to store the (read-only) string $S$ anyway, hence, we restrict our experimental evaluation to this vanilla version.

\subparagraph{Datasets.} We used two real weighted strings which model variations found in the DNA ($|\Sigma| = 4$) of different
samples of the same species. Alternative DNA sequences at a locus are called \emph{alleles}. Alleles have a natural representation as weighted strings: we model the probability $p_i(\alpha)$ in these strings as the relative frequency of the letter $\alpha\in\Sigma$ at the position $i$ among the different samples. We next describe the datasets we used (see also~\cite{DBLP:conf/icde/Gabory0LPZ24}):
\begin{itemize}
    \item \textsf{EFM}: The complete chromosome of the Enterococcus faecium Aus0004 strain (CP003351)\footnote{\url{https://www.ncbi.nlm.nih.gov/nuccore/CP003351}} of length $n=2,955,294$ combined with a set of SNPs\footnote{\url{https://github.com/francesccoll/powerbacgwas/blob/main/data/efm_clade_all.vcf.gz}} taken from $1,432$ samples~\cite{EFM}.
\item \textsf{HUMAN}: The complete chromosome 22 of the Homo sapiens genome (v.~GRCh37)\footnote{\url{https://www.ncbi.nlm.nih.gov/datasets/genome/GCF_000001405.13/}} of length $n=35,194,566$ combined with a set of SNPs\footnote{\url{https://ftp.1000genomes.ebi.ac.uk/vol1/ftp/release/20130502/ALL.chr22.phase3_shapeit2_mvncall_integrated_v5b.20130502.genotypes.vcf.gz}} taken from the final phase of the 1000 Genomes Project (phase 3) representing $2,504$ samples~\cite{1000genomes}.
\end{itemize}
For \textsf{EFM}, we used $z=32$ and for \textsf{HUMAN}, we used $z=8$. We have also used $\ell=\{2^4,\ldots,2^{11}\}$ and then $k=4\log\ell / \log|\Sigma|$ (as typically used~\cite{DBLP:journals/tkde/LoukidesPS23}). 

\subparagraph{Methods.} We have written and considered the following \texttt{C++} implementations:

\begin{itemize}
    \item Heap (\textsf{HP}) approach: Our implementation of the algorithm from \cref{sec:general-queue}.
    \item Sliding-Window (\textsf{SW}) approach: Our implementation of the algorithm from \cref{sec: Sliding window}.
    \item Two-Stack (\textsf{TS}) approach: Our implementation of the algorithm from \cref{sec:2Stack}.
\end{itemize}    
No low-level code optimization has been performed for any of the implementations.
The source code can be freely accessed at \url{https://github.com/wiktor10z/dynamic_minimizers}. 

\subparagraph{Measures.} To measure time, we used the \texttt{std::chrono C++} library. Each experiment was performed three times and the average was recorded (no particular variance was observed).

\subparagraph{Environment.} All experiments were run using a single AMD EPYC 7282 CPU at 2.8GHz with 252GB RAM under GNU/Linux. All methods were implemented in \texttt{C++} and compiled with \texttt{g++} (v.~12.2.1) at optimization level \texttt{-O3}.

\subparagraph{Results.} \cref{fig:EFM} depicts the results for \textsf{EFM} in the two settings discussed. As expected, \textsf{SW} on a trie was not competitive, so its results are omitted. We observe the following:
\begin{itemize}
    \item The runtime of \textsf{SW} and \textsf{TS} (resp.~\textsf{HP}) matches their asymptotic $\Oh(1)$-time (resp.~$\Oh(\log \ell)$-time) bound.
    \item The $\Oh(1)$-time update algorithms (\textsf{SW} and \textsf{TS}) work much faster than the $\Oh(\log \ell)$-time algorithm (\textsf{HP}) even for small values of $\ell$. 
    \item In its worst-case scenario (that is, when the window is moving only one way, the rebuilds occur every $\ell/2$ moves), \textsf{TS} works slower than the sliding-window structure tailored for this application only by a factor smaller than $3.5$.
    \item In the trie application, as $\ell$ grows, 
    the rebuilds of the two-stack data structure get less frequent and \textsf{TS} thus gets (up to a point) faster (while \textsf{HP} becomes slower as expected).
\end{itemize} 
\cref{fig:HUMAN} depicts the results for \textsf{HUMAN} that are analogous. In \cref{app:SARS}, we present results for yet another dataset (SARS-CoV-2) that are also analogous.

\begin{figure}[t]
\begin{subfigure}{0.49\textwidth}
\includegraphics[width=1\linewidth]{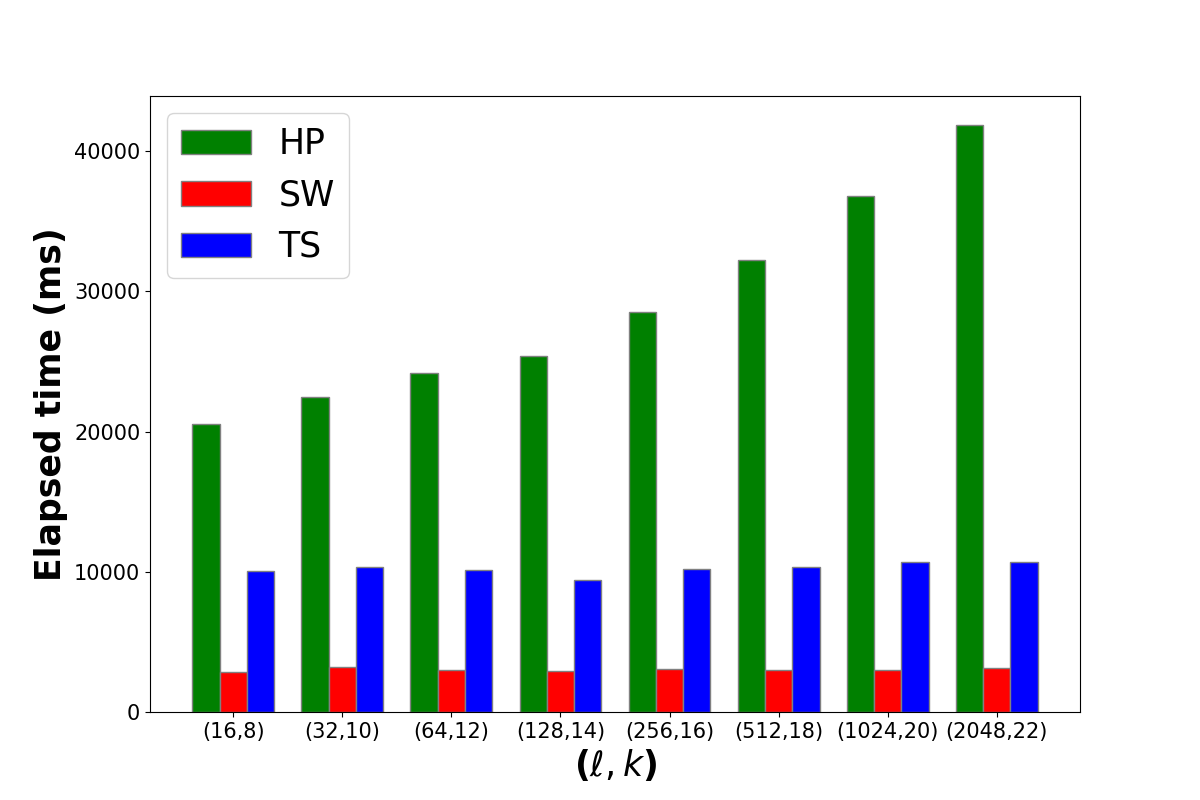} 
\caption{Results on the $z$-estimation.}
\label{fig:subim1}
\end{subfigure}
\begin{subfigure}{0.49\textwidth}
\includegraphics[width=1\linewidth]{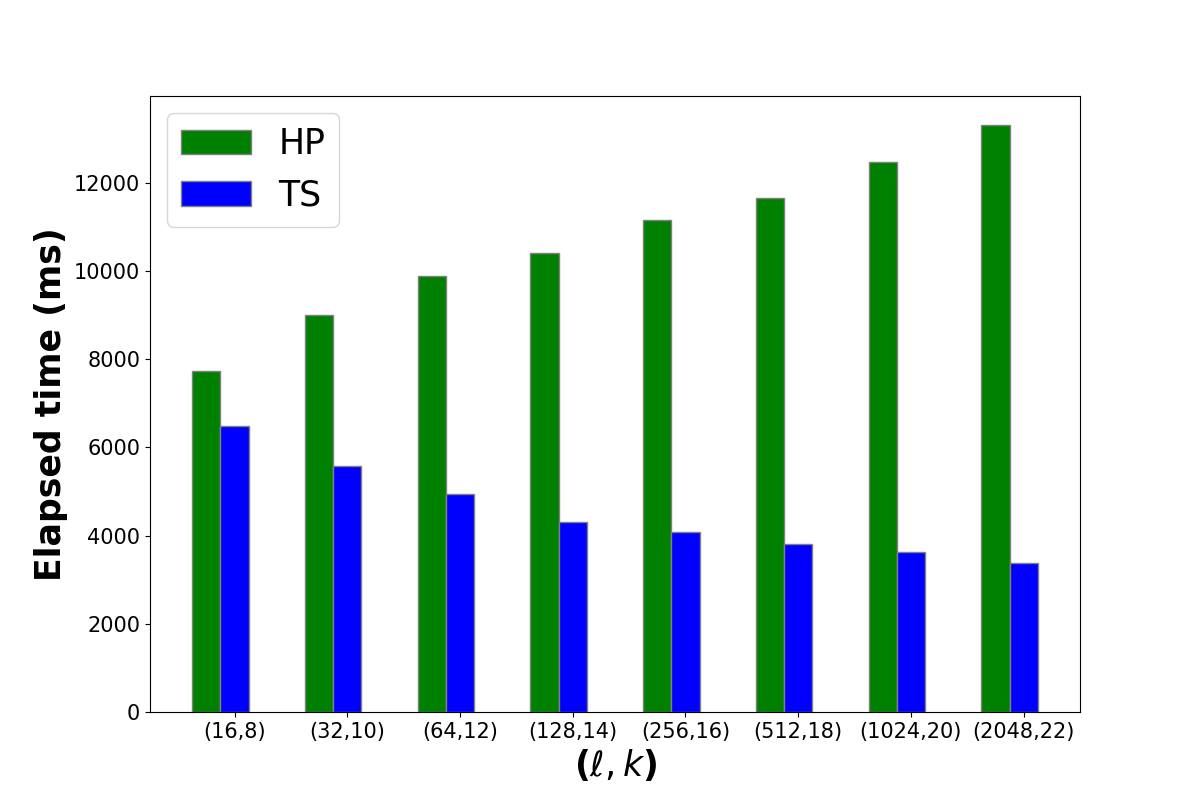}
\caption{Results on the trie.}
\label{fig:subim2}
\end{subfigure}

\caption{The results on the \textsf{EFM} dataset with $z=32$.}
\label{fig:EFM}
\end{figure}

\begin{figure}[t]

\begin{subfigure}{0.49\textwidth}
\includegraphics[width=1\linewidth]{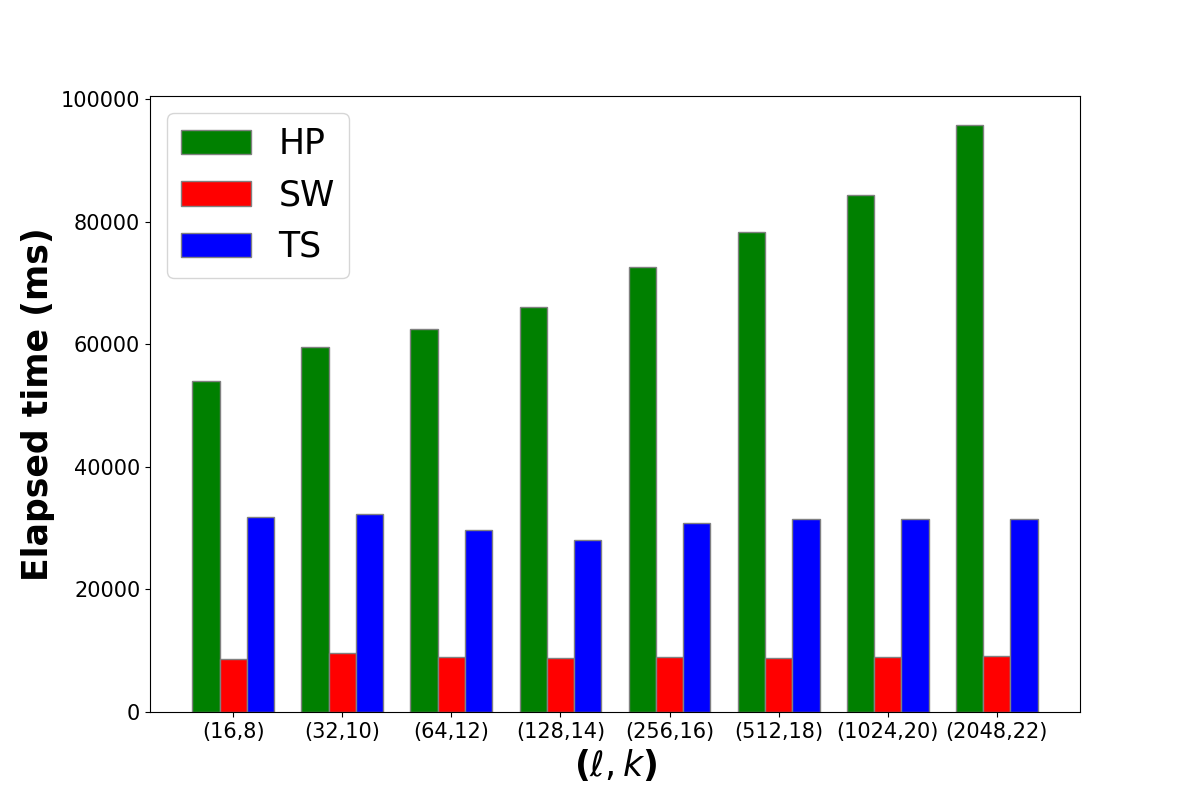} 
\caption{Results on the $z$-estimation.}
\label{fig:subim3}
\end{subfigure}
\begin{subfigure}{0.49\textwidth}
\includegraphics[width=1\linewidth]{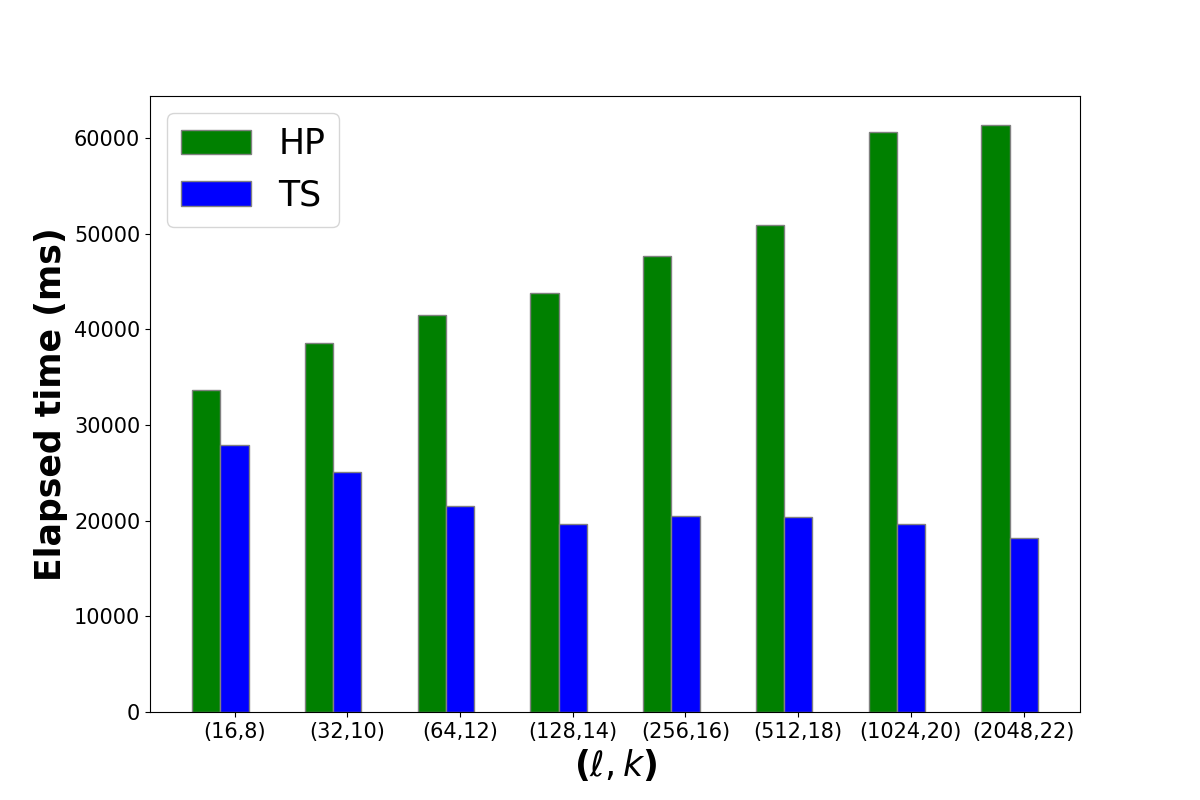}
\caption{Results on the trie.}
\label{fig:subim4}
\end{subfigure}

\caption{The results on the \textsf{HUMAN} dataset with $z=8$.}
\label{fig:HUMAN}
\end{figure}

\bibliographystyle{plain}
\bibliography{references}

\appendix

\section{Omitted Experimental Results}\label{app:SARS}

We have also used the following dataset, which we denote by \textsf{SARS}: the full genome of SARS-CoV-2 (isolate Wuhan-Hu-1)\footnote{\url{https://www.ncbi.nlm.nih.gov/nuccore/MN908947.3}} combined with a set of single
nucleotide polymorphisms (SNPs)\footnote{\url{https://pmc.ncbi.nlm.nih.gov/articles/PMC8363274/\#supp2}} taken from $1,181$ samples~\cite{SARS}.
We have used $z=1024$ for this dataset.

\begin{figure}[ht]
\begin{subfigure}{0.49\textwidth}
\includegraphics[width=1\linewidth]{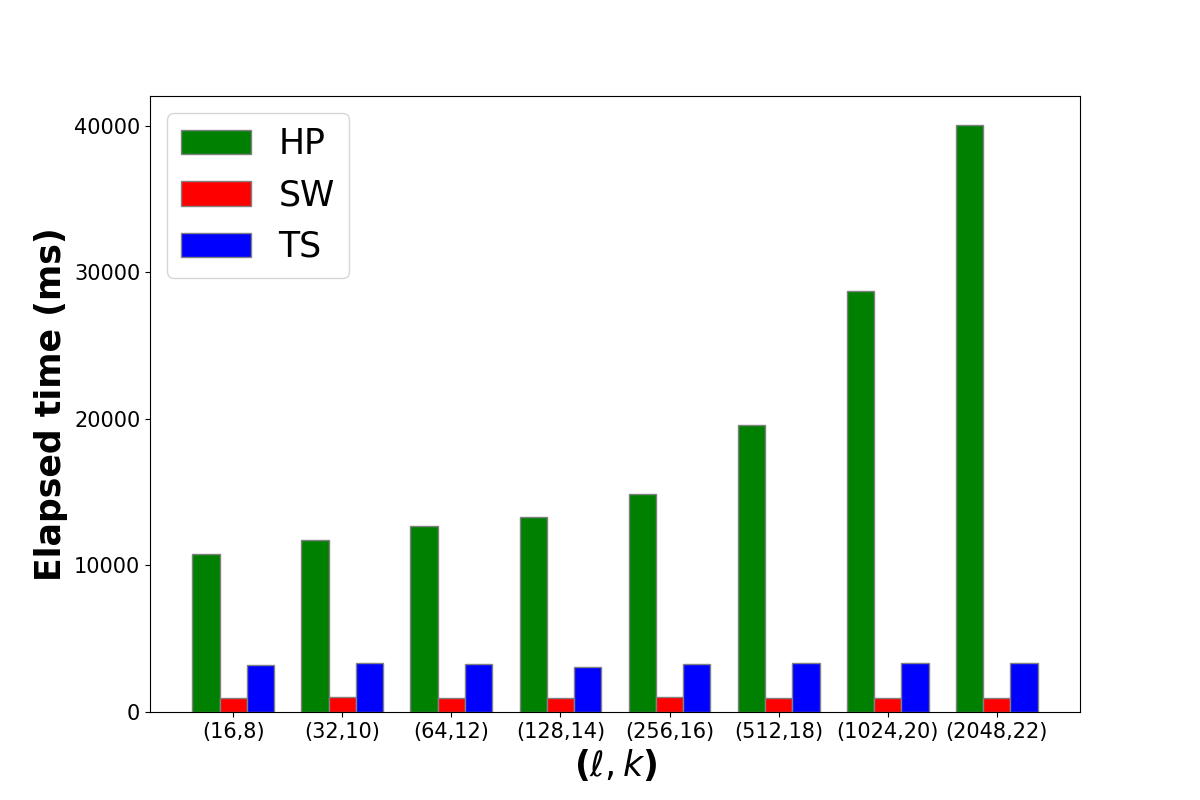} 
\caption{Results on the $z$-estimation.}
\label{fig:subim5}
\end{subfigure}
\begin{subfigure}{0.49\textwidth}
\includegraphics[width=1\linewidth]{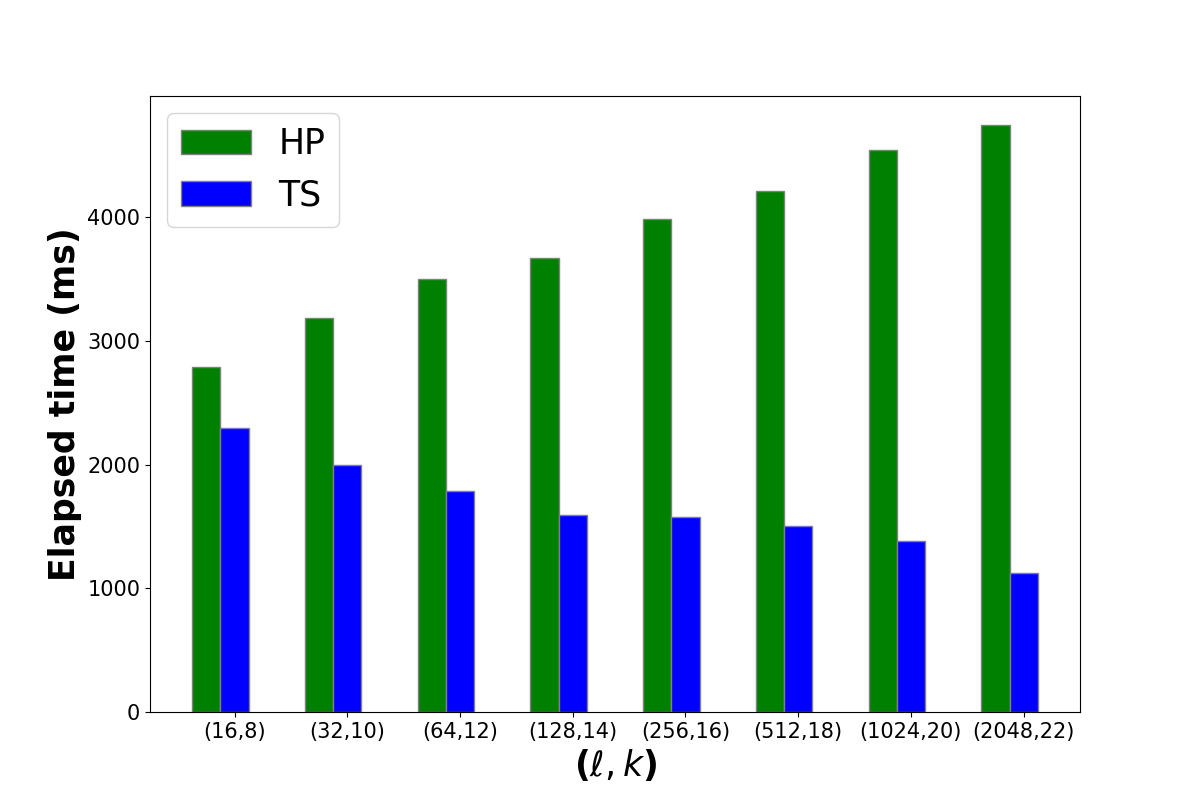}
\caption{Results on the trie.}
\label{fig:subim6}
\end{subfigure}

\caption{The results on the \textsf{SARS} dataset with $z=1024$.}
\label{fig:SARS}
\end{figure}

\end{document}